\documentclass[11pt,reqno]{amsart}
\usepackage{amssymb, amsmath, amsthm,amsrefs,marvosym,wasysym}
\usepackage{latexsym}
\usepackage{color}
\usepackage{exscale}
\usepackage{fullpage}
\usepackage{rotating}
\usepackage{bm, bbm, mathrsfs}
\usepackage{graphics, graphicx}
\newtheorem{theorem}{Theorem}[section]
\newtheorem{lemma}[theorem]{Lemma}
\newtheorem{proposition}[theorem]{Proposition}

\newtheorem{definition}{Definition}[section]
\newtheorem{corollary}[theorem]{Corollary}

\newtheorem{remark}{Remark}[section]

\newtheorem{main_result}{Main Result}[section]

\def\vp{\varphi}

\def\eq#1{(\ref{#1})}
\def\nn{\nonumber}
\def\({\left(\begin{array}{cccccc}}
\def\){\end{array}\right)}

\def\bes{\begin{eqnarray}}
\def\ees{\end{eqnarray}}

\newcommand{\beq}{\begin{equation}}
\newcommand{\eeq}{\end{equation}}
\newcommand{\bea}{\begin{eqnarray}}
\newcommand{\eea}{\end{eqnarray}}
\newcommand{\beann}{\begin{eqnarray*}}
\newcommand{\eeann}{\end{eqnarray*}}

\newcommand{\lam}{\ensuremath{\lambda}}

\newcommand{\RR}{\mathbb{R}}

\newcommand{\br}{{\bar r}}

\DeclareMathOperator{\supp}{supp}
\DeclareMathOperator{\sgn}{sgn}

\newcommand{\s}{\ensuremath{\mathrm{s}}}

\DeclareMathOperator{\grad}{grad}
\DeclareMathOperator{\dv}{div}

\numberwithin{equation}{section}
 
\begin{document}

\title{Radially Symmetric Non-isentropic Euler flows: continuous blowup
with positive pressure}

\begin{abstract}
	Guderley's 1942 work on radial shock waves provides cases of self-similar Euler flows 
	exhibiting blowup of primary (undifferentiated) flow variables: a converging 
	shock wave invades a quiescent region, and the velocity 
	and pressure in its immediate wake become unbounded at time of collapse.
	However, these solutions are of border-line physicality: 
	the pressure vanishes within the quiescent region due to vanishing temperature 
	there. It is reasonable that the lack of upstream counter-pressure
	is conducive to large speeds, with concomitant large amplitudes. Based 
	on Guderley's original solutions it is therefore unclear if it is the zero-pressure 
	region that is responsible for blowup. The same applies to self-similar
	Euler flows describing radial cavity flow, first analyzed by Hunter (1960). 
	
	Recent works have shown that the simplified isothermal and isentropic models 
	admit continuous blowup solutions in the presence of a strictly positive pressure field. 
	In this work we extend this conclusion to the case of the full Euler system. The solutions 
	under consideration are radial self-similar flows in which a continuous wave
	focuses and blows up. We propagate the solutions beyond blowup and 
	observe numerically that there are cases where an expanding spherical shock wave 
	is generated at collapse. The resulting solution has the unusual property
	that the flow is isentropic in each of the two regions separated by the shock. We finally 
	verify that these 
	are admissible global weak solutions to the full, multi-d compressible Euler system.
	\\
	

\end{abstract}

\author{Helge Kristian Jenssen }\address{H.~K.~Jenssen, Department of
Mathematics, Penn State University,
University Park, State College, PA 16802, USA ({\tt
jenssen@math.psu.edu}).}
\author{Charis Tsikkou}\address{C. Tsikkou, School of Mathematical and Data Sciences, West Virginia University,
Morgantown, WV 26506, USA ({\tt
tsikkou@math.wvu.edu}).}

\date{\today}
\maketitle

\tableofcontents

\section{Introduction}\label{intro}
The full (non-isentropic) compressible Euler system expresses conservation of 
mass, linear momentum, and energy in the absence of viscosity and heat conduction \cite{cf}: 
\begin{align}
	\rho_t+\dv_{\bf x}(\rho \bf u)&=0 \label{mass_m_d_full_eul}\\
	(\rho{\bf  u})_t+\dv_{\bf x}[\rho {\bf  u}\otimes{\bf  u}]+\grad_{\bf  x} p&=0
	\label{mom_m_d_full_eul}\\
	(\rho E)_t+\dv_{\bf x}[(\rho E+p){\bf  u}]&=0.\label{energy_m_d_full_eul}
\end{align}
The independent variables are time $t$ and position ${\bf  x}\in\RR^n$, and the 
dependent variables are density $\rho$, fluid velocity ${\bf  u}$, 
and internal energy $e$; the total energy density is $E=e+\textstyle\frac{1}{2}|{\bf  u}|^2$. 
We restrict attention to ideal gases with pressure $p$ given by
\beq\label{pressure1}
	p(\rho,e)=(\gamma-1)\rho e\qquad\qquad\text{($\gamma>1$ constant)}.
\eeq
Assuming constant specific heat, the
internal energy is proportional to the temperature $\theta$ of the gas: $e\propto\theta$.
The local speed of sound $c$ is given by
\beq\label{sound_speed}
	c=\sqrt{\textstyle\frac{\gamma p}{\rho}}=\sqrt{\gamma(\gamma-1)e}\propto\sqrt\theta.
\eeq
We refer to $\rho$, $\bf u$, $p$, $c$, and $\theta$ as primary (undifferentiated) 
flow variables.

The present work deals with the phenomenon of {\em amplitude blowup}
in multi-d Euler flow. The simplest setting for this is via 
radial flows, i.e., flows in which the variables depend on
position only through $r=|{\bf x}|$, and the 
velocity field is purely radial, viz.\ ${\bf u}=u\frac{\bf x}{r}$. 
With these assumptions, and within smooth regions of the flow, 
\eq{mass_m_d_full_eul}-\eq{energy_m_d_full_eul} take the form
\begin{align}
	\rho_t+u\rho_r+\rho(u_r+\textstyle\frac{mu}{r}) &= 0\label{m_eul}\\
	u_t+ uu_r +\textstyle\frac{1}{\gamma\rho}(\rho c^2)_r&= 0\label{mom_eul}\\
	c_t+uc_r+{\textstyle\frac{\gamma-1}{2}}c(u_r+\textstyle\frac{mu}{r})&=0\label{ener_eul}
\end{align}
where $\rho=\rho(t,r)$, $u=u(t,r)$, $c=c(t,r)$, and $m=n-1$. 

The presence (when $n>1$) of terms with the unbounded geometric factor $\frac{1}{r}$
makes it reasonable that Euler flows can generate infinite amplitudes in 
a radially converging flow. It is natural to ask if there are mechanisms, within the Euler 
model itself, that could prevent such behavior. One candidate would be 
that the fluid near the center of motion provides a counter-pressure which would 
slow down the incoming flow, averting actual blowup.

It turns out that the Euler-system {\em does} admit solutions with infinite
amplitudes. These flows are commonly named after Guderley who was the 
first to study them in \cite{gud}. They have received much attention in the fluid 
dynamics literature, especially in connection to Inertial Confinement Fusion,
see \cites{am,dm,p}. These ``Guderley solutions'' are self-similar solutions (to be described below) of 
the non-isentropic Euler system \eq{mass_m_d_full_eul}-\eq{energy_m_d_full_eul}
for an ideal, polytropic gas. On the other hand, and regardless of their blowup behavior,
they are also at the limit of applicability of the Euler model.

To see why, we recall that Guderley solutions involve a single focusing shock 
wave which moves toward the origin by invading a fluid at rest and at 
constant density and pressure, leaving a dynamically changing flow in its wake \cites{laz,jt1}.
The incoming shock accelerates up to the time of ``collapse,'' hits the origin with infinite speed, 
and is instantaneously reflected into an expanding shock wave. The reflected shock proceeds 
outward as it interacts with the still-incoming flow generated by the original, converging shock. 
At collapse, the velocity, pressure, sound speed 
(and hence also temperature) all attain infinite values at the center of motion.
In contrast, the density is everywhere constant and finite at collapse in Guderley's solutions.

As emphasized by Lazarus (see p.\ 318 in \cite{laz}), for the quiescent state 
within the converging shock to be an exact self-similar solution
of the Euler system, the sound speed there must vanish identically. 
For an ideal gas this implies vanishing pressure and it follows that, in a 
Guderley solution, the fluid near the origin offers no counter-pressure 
for the incoming shock. We stress that the density field within the quiescent region 
is strictly positive. Thus, the vanishing pressure
there is due to a vanishing temperature (``cold gas'') and not a vacuum. 
The flows exhibited by these solutions are therefore of a borderline type 
where the physicality of the Euler model is open to question.
(On the other hand, they do provide genuine global weak solutions to the 
full multi-d Euler system; \cite{jt1}.)

A distinct case of amplitude blowup 
for the Euler system occurs for radial self-similar flow describing a collapsing
cavity, i.e., a spherical vacuum region being filled by an inflowing gas. This type of
solution was first described by Hunter \cite{hun_60}; see also Lazarus \cite{laz} who analyzed 
both cases of self-similar flows in a combined treatment. Thus, in both types of
blowup solution, the unbounded growth occurs in the immediate wake of a wave
which enters a region at zero pressure: in Guderley solutions a shock moves
into a ``cold gas,'' while for cavity flow a free interface expands into a vacuum.
It is reasonable to ask if it is  precisely the lack of a positive counter-pressure 
which allows for this type of blowup behavior.

Recent results on self-similar blowup for the simplified isothermal and isentropic Euler 
models \cites{jt2,jt3,mrrs1,mrrs2,b_c-l_g_s} indicate that the answer is ``No.''
These works provide examples of radially converging isothermal and isentropic flows 
which suffer amplitude blowup in the presence of a strictly positive pressure field. 
{\em The first goal of the present work is to extend this conclusion to the 
full Euler system \eq{mass_m_d_full_eul}-\eq{energy_m_d_full_eul}},
demonstrating that the effect of wave focusing is 
sufficiently strong on its own to generate unbounded values
of primary flow variables.

%
%
%

A common ingredient in \cites{jt2,jt3,mrrs1,mrrs2,b_c-l_g_s}, as well as the 
present work, is the use of a converging wave which is {\em continuous} prior to collapse,
in contrast to Guderley's converging shock solution.

\begin{remark}
        In fact, the works \cites{mrrs1,mrrs2,b_c-l_g_s} construct {\em smooth} 
        self-similar isentropic Euler flows, which are then used to give examples of 
        nearby blowup solutions of the Navier-Stokes system.
        The $C^\infty$ regularity of the underlying Euler flows is critical for the 
        constructions in \cites{mrrs1,mrrs2,b_c-l_g_s}. 
        For our purpose of constructing global-in-time solutions to the 
        full Euler system \eq{mass_m_d_full_eul}-\eq{energy_m_d_full_eul}
        with a strictly positive pressure, the exact regularity of 
        the solution prior to collapse is not essential.
\end{remark}


A key observation in the present work is that, once we settle on the use of a continuous, 
radial self-similar wave to generate blowup for \eq{m_eul}-\eq{ener_eul}, 
certain physical constraints imply that the incoming flow is necessarily isentropic.
More precisely, we show that the requirements of
\begin{enumerate}
	\item locally finite amounts of mass, momentum, and energy at all times; and 
	\item no unbounded amplitudes of primary flow variables occur prior to collapse, 
\end{enumerate}
together imply a constant value of the entropy field in the incoming flow; cf.\ Proposition 
\ref{kappa_value}. (Of course, this 
is under the assumption of radial self-similar flow.)
As a consequence, we can make use of the analysis in \cites{jt3,laz} to build the 
incoming flow. Another consequence of this is that {\em all} the primary flow variables,
including the density (in contrast to the case of Guderley solutions), suffer blowup 
at time of collapse. 

{\em Our second goal is to continue the resulting solutions beyond collapse and 
obtain global-in-space-and-time, admissible weak solutions to the full Euler system 
\eq{m_eul}-\eq{ener_eul}}, an issue not addressed in \cites{mrrs1,mrrs2,b_c-l_g_s}. 
For this we have found it necessary to rely on (robust) numerical evidence that
there are cases where the infinite 
amplitudes in primary flow variables result in the generation of an outgoing self-similar
shock wave. We formulate this as a technical condition (see Condition $(\Pi)$ in Section \ref{t>0})
which requires that the Hugoniot locus of a certain trajectory of the reduced similarity 
ODE \eq{CV_ode} intersects a certain other trajectory.

\begin{remark}
	Figure \ref{Figure_2} (generated with Maple) provides numerical 
	evidence in a representative case: Condition $(\Pi)$ 
	requires that the Hugoniot locus $\Gamma_H$ of trajectory $\Gamma_3$ 
	intersects the trajectory $\Gamma_4$. The jump from the lower point on $\Gamma_3$ to
	the point $P_H\in\Gamma_4$ corresponds to the outgoing shock wave generated at collapse.
	
	We mention another, somewhat more subtle, issue. Namely, numerical tests 
	suggest the existence of a different type of solution: at least for some choices of $n=2,3$,
	$\gamma>1$, and $\lambda>1$ as described by the Main Result, it is possible 
	to build radial similarity solutions that are {\em globally continuous}, save for a one-point
	amplitude 	blowup at time of collapse (see Remark \ref{Pi2}). I.e., there is clear evidence  
	that the presence of unbounded amplitudes of the primary flow variables 
	$(\rho,u,c,p,\theta)$,  at a single point in space-time, 
	does not necessarily lead to shock formation at that point. 
	Indeed, the unlabeled trajectory joining $P_8$ to $P_9$ in Figure \ref{Figure_2}
	corresponds to this type of behavior. A related scenario 
	was recently observed in \cite{jj_22} for locally bounded Euler flows: gradient blowup 
	of primary flow variables does not necessarily lead to shock formation.
	A rigorous analytic justification of these behaviors appears to be challenging 
	and will be pursued elsewhere.
\end{remark}

We note that the physical constraints (1) and (2) above have a further 
consequence which highlights the highly special form of the solutions under consideration: 
not only is the incoming flow field isentropic, but the same applies
to the flow field in the wake of the expanding shock wave generated at collapse.
In particular, the specific entropy takes only two values: one in the region outside the 
expanding shock ($\equiv$ its value in the original, incoming flow), and one in the 
region within the shock.
(As far as we are aware, the only other known example of this behavior is 
that of a constant-speed, planar shock wave connecting two constant states.)
However, we stress that the flows we construct are solutions to the {\em full}
Euler system. Specifically, the two constant values of the specific entropy are necessarily 
distinct, and the solution is not globally isentropic;
see Remark \ref{isentropic_comparison}.

Finally, having obtained globally defined self-similar solutions to \eq{m_eul}-\eq{ener_eul}, 
it remains to verify that they provide genuine weak solutions to the original,
multi-d Euler system \eq{mass_m_d_full_eul}-\eq{energy_m_d_full_eul}. 
Due to the amplitude blowup at the origin, this is not immediate.
For this we require that the conserved variables (mass, momentum, energy) maps time 
continuously into $L^1_{loc}(\RR^n)$, that all terms occurring in the weak form of 
the full, multi-d Euler system are locally integrable in space-time, and 
that the weak forms of \eq{mass_m_d_full_eul}-\eq{energy_m_d_full_eul} 
are satisfied for all test function belonging to $C_c^1(\RR_t\times\RR^n_{\bf x})$
(cf.\ Definition \ref{weak_soln}).
In particular, verification of the weak form of the energy equation requires
estimation of additional, higher order terms that are not present for the isothermal and 
isentropic Euler models.

%
%
%
%

\subsection{Main result and outline}
In stating the main result we make use of terminology that is detailed 
in Section \ref{setup} below. (We abuse notation slightly
in using the same symbols for the density and energy fields in rectangular 
and radial coordinates.) Note that the time of collapse is chosen to be $t=0$.
\begin{main_result}\label{main_result}
	Consider the full, multi-d Euler system 
	\eq{mass_m_d_full_eul}-\eq{energy_m_d_full_eul} for an ideal,
	polytropic gas in space dimension $n=2$ or $n=3$, and consider 
	radial self-similar solutions 
	\beq\label{1_to_multi-d}
		\rho(t,{\bf x}):=\rho(t,|{\bf x}|)\qquad
		{\bf u}(t,{\bf x}):=u(t,|{\bf x}|)\textstyle\frac{{\bf x}}{|{\bf x}|},\qquad
		E(t,{\bf x}):=E(t,|{\bf x}|)\
	\eeq
	given via \eq{alt_sim_vars}, with $V$ and $C$ solving the similarity ODEs \eq{V_sim2}-\eq{C_sim2},
	and $R$ given by the exact integral in \eq{entr_int}. 
	Then, for each $\gamma>1$ there is a number $\hat\lambda=\hat\lambda(\gamma,n)>1$ 
	such that the following holds. 
	\begin{enumerate}
		\item[(a)] For each similarity exponent $\lambda\in(1,\hat\lambda)$ 
		there are infinitely many continuous, radial self-similar solutions $(\rho,u,E)(t,r)$
		defined on $\RR_t^-\times\RR_r^+$ which suffer amplitude blowup in primary 
		flow-variables at the origin as $t\uparrow0$. 
		\item[(b)] Furthermore, each solution satisfying Condition $(\Pi)$ (see Section 
		\ref{t>0}) can be continued 
		to all times $t>0$ and contains an admissible, expanding shock wave emanating from the origin. 
		The resulting triples $(\rho(t,{\bf x}),{\bf u}(t,{\bf x}),E(t,{\bf x}))$ defined by \eq{1_to_multi-d}
		are globally defined and provide admissible weak solutions of the full multi-d Euler system 
		\eq{mass_m_d_full_eul}-\eq{energy_m_d_full_eul} according to Definition \ref{weak_soln}.
		Finally, the pressure field in these solutions is everywhere strictly positive at all times.
	\end{enumerate}
\end{main_result}
The rest of the paper is organized as follows.  
In Section \ref{setup} we follow \cites{cf,laz} in setting up the
framework for radial self-similar solutions to \eq{mass_m_d_full_eul}-\eq{energy_m_d_full_eul}.
This includes a discussion of the similarity ODEs \eq{V_sim2}-\eq{C_sim2} and the particular solutions of 
these that we make use of. The analysis in Section \ref{lam_kap_constrs} shows 
how the physical requirements (1)-(2) above single out a unique value of the similarity parameter 
$\kappa$, which turns out  to characterize isentropic flow. This allows us to obtain 
the relevant self-similar solutions as in the isentropic case, the details of which were 
given in \cite{jt3}; the argument is outlined in Section \ref{constr}. Finally, in Section 
\ref{weak_solns} we establish that the resulting flows are genuine global weak solution of
the full multi-d Euler system \eq{mass_m_d_full_eul}-\eq{energy_m_d_full_eul}.

\section{Similarity variables and similarity ODEs}\label{setup}
We start by prescribing the form of similarity solutions to be 
used in the rest of the paper; see \cites{gud,cf,sed,stan,rj}. By invariance under time translation
we are free to choose $t=0$ to be the time of collapse, and this choice is 
built into the choice of the similarity variable $x$. 
Following \cites{laz,cf} we use the variables
\beq\label{alt_sim_vars}
	x=\frac{t}{r^\lambda},\qquad \rho(t,r)=r^\kappa R(x),\qquad 
	u(t,r)=-\frac{r^{1-\lambda}}{\lambda}\frac{V(x)}{x}, \qquad 
	c(t,r)=-\frac{r^{1-\lambda}}{\lambda}\frac{C(x)}{x}.
\eeq
At this stage the parameters $\lambda$ and $\kappa$ are free. 
However, we are interested in solutions suffering blowup at $t=0$, which 
corresponds to $x=0$. The similarity solutions to be built will be such that 
$\frac{V(x)}{x}$ and $\frac{C(x)}{x}$ approach finite limits as $x\to 0$, so that,
e.g., $u(0,r)\propto r^{1-\lambda}$ according to \eq{alt_sim_vars}. 
Consequently, we shall focus exclusively on similarity exponents satisfying 
$\lambda>1$. Further constraints are imposed below.

\begin{remark}[The parameter $\kappa$]
	For standard Guderley solutions (see Section \ref{intro}) the incoming shock penetrates a quiescent 
	region near $r=0$. The constancy of the density field $\rho(t,r)=r^\kappa R(x)$ 
	there implies that $\kappa$ must vanish. In contrast, for 
	the solutions we consider in the present work the density will change dynamically
	in all of space-time, and $\kappa$ is therefore not determined a priori.
	However, we shall show that the requirements $\mathrm{(1)}$-$\mathrm{(2)}$ in Section \ref{intro}
	in fact determine $\kappa$ uniquely in terms of $\lambda$ and $\gamma$; 
	see \eq{isentr_kappa} and Remark \ref{isentropic_comparison}. 
\end{remark}
Substitution of \eq{alt_sim_vars} into \eq{m_eul}-\eq{mom_eul} yield three coupled, 
non-linear ODEs, the {\em similarity ODEs}, for $R(x)$, $V(x)$, $C(x)$.
It is a remarkable fact (first observed by Guderley in the case $\kappa=0$, according to \cite{cf}) 
that $R$ can be eliminated to give two coupled ODEs for only $V$ and $C$, viz.
\begin{align}
	\frac{dV}{dx}&=-\frac{1}{\lambda x}\frac{G(V,C)}{D(V,C)}\label{V_sim2}\\
	\frac{dC}{dx}&=-\frac{1}{\lambda x}\frac{F(V,C)}{D(V,C)},\label{C_sim2}
\end{align}
which in turn yield a single, autonomous ODE 
(no explicit $x$-dependence) 
\beq\label{CV_ode}
	\qquad\qquad\qquad\frac{dC}{dV}=\frac{F(V,C)}{G(V,C)} \qquad\qquad\text{(reduced similarity ODE)}
\eeq
relating $V$ and $C$ along similarity solutions. A direct calculation shows that 
the functions $D$, $F$, $G$ are given by
\begin{align}
	D(V,C)&=(1+V)^2-C^2\label{D}\\
	G(V,C)&=nC^2(V-V_*)-V(1+V)(\lam+V)\label{G}\\
	F(V,C)&=C\left\{C^2\big(1+\textstyle\frac{\alpha}{1+V}\big)
	-k_1(1+V)^2+k_2(1+V)-k_3\right\},\label{F}
\end{align}
where
\beq\label{V_*}
	V_*=\textstyle\frac{\kappa-2(\lambda-1)}{n\gamma},
\eeq
\beq\label{alpha}
	\alpha=\textstyle\frac{1}{\gamma}[(\lambda-1)+\frac{\kappa}{2}(\gamma-1)],
\eeq
and
\beq\label{ks}
	k_1=1+{\textstyle\frac{(n-1)(\gamma-1)}{2}},\qquad 
	k_2={\textstyle\frac{(n-1)(\gamma-1)+(\gamma-3)(\lam-1)}{2}},\qquad
	k_3=\textstyle\frac{(\gamma-1)(\lam-1)}{2}.
\eeq
Note that $V=V_*$ is a vertical asymptote for the zero-level set of $G$ in the $(V,C)$-plane.
For later reference we introduce the null-clines
\[\mathcal F:=\{(V,C)\,|\,F(V,C)=0\}\qquad\mathcal G:=\{(V,C)\,|\,G(V,C)=0\},\]
and the ``critical lines''
\[L_\pm:=\{(V,C)\,|\,C=\pm(1+V)\}.\]
We note the symmetries
\beq\label{symms}
	G(V,-C)=G(V,C),\qquad F(V,-C)=-F(V,C).
\eeq
The autonomous ODE \eq{CV_ode} plays a key role in the analysis:
Certain of its trajectories will provide the various parts of the flows we seek.
This requires a detailed analysis of its phase portrait in the $(V,C)$-plane,
including the dependencies on the parameters $n$, $\gamma$, 
$\lambda$, and $\kappa$. We shall refer to \cites{jt3,laz} for the relevant 
facts. 

To fully describe the physical flow we also need to recover the density 
field $\rho(t,r)=r^\kappa R(x)$. This is obtained from $V(x)$ and $C(x)$ via 
an exact integral, the 
so-called ``entropy integral'' \cites{rj,laz}, for the similarity ODEs:
\beq\label{entr_int}
	\textstyle\left(\frac{C}{x}\right)^2\!R^{1-\gamma}[R|1+V|]^q\equiv S\qquad\text{($S$ constant $>0$)},
\eeq  
with $R\geq 0$ and  
\beq\label{q}
	q=\textstyle\frac{1}{\kappa+n}[\kappa(\gamma-1)+2(\lambda-1)]
	\equiv \frac{2\gamma}{\kappa+n}\alpha.
\eeq 
The existence of this exact integral amounts to the fact that the specific entropy remains constant 
along particle paths in smooth Euler flows. The constant on the right-hand side 
of \eq{entr_int} takes different values in each region of smoothness. 

It turns out that \eq{CV_ode} has up to 11 critical points (including two at infinity); 
however, not all of these are relevant for building
physically meaningful radial Euler flows that suffer amplitude blowup. (We follow the 
labeling in \cite{laz}.)
Among the ones we will make use of are 
\beq\label{p+-}
	P_{\pm\infty}=(V_*,\pm\infty), \qquad\text{where $V_*$ is given in \eq{V_*}.}
\eeq
In particular, the relevant solutions of \eq{CV_ode} will approach the critical points
$P_{\pm\infty}$ as $x\to\mp\infty$, respectively. 
To analyze these we focus on $P_{+\infty}$ (sufficient according to \eq{symms})
and change to the variables $W=V-V_*$ and $Z=C^{-2}$, so that $P_{+\infty}$ 
corresponds to the origin in the $(W,Z)$-plane. Linearizing the 
resulting equation for $dZ/dW$ about $(W,Z)=(0,0)$ gives the ODE
\beq\label{WZ_ode}
	\frac{dZ}{dW}=-\frac{aZ}{nW-bZ},
\eeq
where 
\beq\label{ab}
	a=2(1+\textstyle\frac{\alpha}{1+V_*}),\qquad b=V_*(1+V_*)(\lambda+V_*).
\eeq
We shall want $P_{+\infty}$ to be a saddle point, which imposes the constraint $a>0$; this requirement 
is addressed below in Section \ref{lam_kap_constrs}.

The Euler solutions we construct are built from four trajectories $\Gamma_1-\Gamma_4$
of \eq{CV_ode} (see Figure 1):
\begin{itemize}
	\item $\Gamma_1$ connects $P_{+\infty}$ to a critical point $P_8$ in the 2nd quadrant 
	of the $(V,C)$-plane. Passing through the 
	point $P_8$ corresponds to crossing the ``critical characteristic'' or ``sonic curve'' in the $(r,t)$-plane, i.e., the 
	1-characteristic which reaches the center of motion at time of collapse.
	\item $\Gamma_2$ connects $P_8$ to the origin $P_1=(0,0)$, which turns out to be a 
	star point for \eq{CV_ode}. There is an infinite number of such trajectories and we are 
	free to let $\Gamma_2$ be any one of these, the only constraint 
	being that it reaches the origin with a strictly negative and finite slope.
	\item $\Gamma_3$ is the trajectory that moves into the 4th quadrant with the same slope that
	$\Gamma_2$ arrived with at the origin; it is continued until a certain point (to be determined),
	at which it jumps to the corresponding ``Hugoniot point" $P_H$ (located in the 3rd quadrant in Figure 1). 
	This jump corresponds to the outgoing shock generated at the center of motion in physical space
	at time $t=0$.
	\item $\Gamma_4$ connects $P_H$ to $P_{-\infty}$; due to the symmetries \eq{symms},
	the trajectory $\Gamma_4$ coincides with a part of the reflection of $\Gamma_1$ about
	the $V$-axis.
\end{itemize}

The challenge is to verify that for fixed values of $n=2$ or $3$ and $\gamma>1$, there are 
values $\lambda>1$ and $\kappa$ for which such trajectories exist, and that the resulting
Euler solutions are genuine weak solutions. We establish the first part by showing that the 
requirements (1) and (2) in Section \ref{intro} reduce the analysis to the isentropic setting
treated in \cite{jt3} (Sections \ref{lam_kap_constrs}-\ref{constr} below). The second part is 
addressed in Section \ref{weak_solns}.

\section{Restrictions on $\lambda$ and $\kappa$}\label{lambda_restrict}\label{lam_kap_constrs}
We assume throughout that $\lambda>1$ and $\gamma>1$.
We shall construct solutions $(V(x),C(x))$ of \eq{V_sim2}-\eq{C_sim2}
 with the property that 
\beq\label{condition1}
	\frac{V(x)}{x}\to \nu\qquad\text{and}\qquad \frac{C(x)}{x}\to \mu\qquad\text{as $x\to0$,}
\eeq
where $\nu>0$ and $\mu<0$ are finite numbers. In particular, this will ensure 
that the trajectory denoted $\Gamma_2$ above reaches the origin in the 
$(V,C)$-plane with the finite, negative slope $\frac{\mu}{\nu}$. It also implies,
via \eq{alt_sim_vars}, that the flow variables at time of collapse are given by
\beq\label{at_collapse}
	\rho(0,r)=R(0)r^\kappa\qquad 
	u(0,r)=-\frac{\nu}{\lambda}r^{1-\lambda}, \qquad 
	c(0,r)=-\frac{\mu}{\lambda}r^{1-\lambda}.
\eeq
As $\lambda>1$, it follows that the particle and sound speeds both 
blow up at the center of motion at time collapse in the solutions 
we consider; we shall see that the same applies to the density.

Next, we insist that the flows under consideration are physical in the 
sense that they contain bounded amounts of mass, momentum, 
and energy within a fixed ball about the origin:
\[\int_0^1\rho(t,r)r^m\, dr,\quad
\int_0^1\rho(t,r)|u(t,r)|r^m\, dr,\quad
\int_0^1\rho(t,r)\left(e(t,r)+\textstyle\frac{1}{2}|u(t,r)|^2\right)r^m\, dr<\infty,
\]
By using \eq{at_collapse} it is straightforward to verify that,
at time $t=0$, these integrability constraints amount to the following conditions:
\begin{itemize}
	\item[(I)] $\kappa+n>0$
	\item[(II)] $\lambda<1+\kappa+n$
	\item[(III)] $\lambda<1+\textstyle\frac{\kappa+n}{2}$,
\end{itemize}
respectively. Note that (II) is a consequence of (I) and (III).
We record a consequence of these conditions:
According to (III) and the standing assumption $\gamma>1$, we have
\beq\label{1+V*}
	0<\textstyle\frac{n+\kappa-2(\lambda-1)}{n\gamma}
	<\textstyle\frac{n\gamma+\kappa-2(\lambda-1)}{n\gamma}\equiv 1+V_*,
\eeq
where $V_*$ is defined in \eq{V_*}.
Next, as the main goal is to establish existence of radial similarity 
solutions of the Euler system that suffer amplitude blowup at time $t=0$
(in the presence of an everywhere positive pressure field), we also insist that
the primary flow variables $\rho$, $u$, and $c$ (or, equivalently by \eq{sound_speed}, $\theta$) 
remain bounded at any 
fixed time $\bar t<0$ prior to collapse. In particular, these quantities should
remain bounded as $r\downarrow 0$. First, consider 
\beq\label{u_t_bar}
	u(\bar t,r)=-\frac{r^{1-\lambda}}{\lambda}\frac{V(x)}{x}=-\frac{1}{\lambda \bar t}V(x)r
	\propto V(x)r. 
\eeq
Recalling that our solutions will be constructed so that $(V(x),C(x))$ 
tends to $P_{+\infty}=(V_*,+\infty)$ as $x\downarrow -\infty$,
we obtain from \eq{u_t_bar} that $u(\bar t,r)\sim r$ as $r\downarrow 0$. This shows that,
at any time prior to collapse, the speed of the fluid 
particles approaches zero at a linear rate as the center of motion is approached.
Thus, no additional constraint is imposed by requiring boundedness (indeed, vanishing) 
of the fluid speed near the center of motion. 

Next, to analyze $c(\bar t,r)$ as $r\downarrow 0$, we need the leading order behavior 
of $C(x)$ as $x\downarrow -\infty$. Recalling that $(V(x),C(x))\to(V_*,+\infty)$ as $x\downarrow-\infty$,
we obtain from \eq{C_sim2} that
\[\textstyle\frac{1}{C}\frac{d C}{dx}\sim \frac{1}{\lambda}(1+\frac{\alpha}{1+V_*}) \frac{1}{x} 
\qquad\text{as $x\downarrow -\infty$.}\]
It follows that 
\beq\label{sigma}
	C(x)\sim |x|^\sigma\qquad\text{as $x\downarrow -\infty$, where} \qquad
	\sigma=\textstyle\frac{1}{\lambda}(1+\frac{\alpha}{1+V_*}).
\eeq
For $\bar t<0$ fixed, we have $x\propto- r^{-\lambda}$, so that \eq{alt_sim_vars}${}_4$ gives
\beq\label{c_bar}
	c(\bar t,r)\sim r^{1-\sigma\lambda}\qquad\text{as $r\downarrow 0$.}
\eeq
Therefore, to have $c(\bar t,r)$ bounded near the center of motion prior to 
collapse, we must require $1-\sigma\lambda\geq0$, which, by
\eq{sigma} and \eq{1+V*}, amounts to 
\beq\label{alfa}
	\alpha\leq 0.
\eeq
We also note that (I) yields 
$n(\gamma-1)-2(\lambda-1)>-\kappa(\gamma-1)-2(\lambda-1)\equiv -2\gamma\alpha$, so that 
\eq{alfa} gives
\beq\label{cond_1}
	[n(\gamma-1)-2(\lambda-1)]\alpha\leq -2\gamma\alpha^2.
\eeq
To obtain the behavior of  $\rho(\bar t,r)$ as $r\downarrow 0$, we use the exact integral
\eq{entr_int}, together with \eq{1+V*}, $V(x)\sim V_*$, $C(x)\sim |x|^\sigma$, and $x\propto r^{-\lambda}$, 
to get that
\beq\label{rho_bar}
	\rho(\bar t,r)\sim r^{\kappa+\frac{2\lambda(\sigma-1)}{1-\gamma+q}}
	\qquad\text{as $r\downarrow 0$,}
\eeq
where $q$ is given by \eq{q}. To have $\rho(\bar t,r)$ bounded as $r\downarrow0$, we must 
therefore have 
\beq\label{q_etc}
	\kappa+\textstyle\frac{2\lambda(\sigma-1)}{1-\gamma+q}\geq 0.
\eeq
Observe that, by \eq{q}, (I), and \eq{alfa}, we have $q\leq 0$. Therefore,
the denominator in \eq{q_etc} satisfies $1-\gamma+q\leq 1-\gamma<0$,
so that \eq{q_etc} amounts to
\[\kappa(1-\gamma+q)+2\lambda(\sigma-1)\leq 0.\]
Using \eq{q} (first equality) and \eq{sigma} to substitute for $q$ and $\sigma$, respectively, and 
rearranging, we obtain the equivalent condition 
\beq\label{almost_there}
	\textstyle\frac{\alpha}{1+V_*}\leq 
	\frac{n}{\kappa+n}[(\lambda-1)+\frac{\kappa}{2}(\gamma-1)]
	\equiv\frac{n\gamma \alpha}{\kappa+n}.
\eeq
Again by (I) and \eq{1+V*}, the last condition is equivalent to
\[[n\gamma(1+V_*)-(\kappa +n)]\alpha\geq 0,\]
or, substituting for $V_*$ from \eq{V_*},
\beq\label{cond_2}
	[n(\gamma-1)-2(\lambda-1)]\alpha\geq0.
\eeq
Combining \eq{cond_1} and \eq{cond_2} yields $0\leq-2\gamma\alpha^2$. As $\gamma>1$ by
assumption, we conclude that $\alpha$ must vanish: $\alpha=0$. This condition
uniquely determines $\kappa$ in terms of $\gamma$ and $\lambda$; see \eq{isentr_kappa}.
Finally, it is immediate to verify from \eq{c_bar} and \eq{rho_bar} that this ``isentropic'' value for $\kappa$ 
(see Remark \ref{isentropic_comparison} below)
indeed yields bounded values for $\rho(\bar t,r)$ and $c(\bar t,r)$ as $r\downarrow0$.
Summing up, we have established the following. 

\begin{proposition}\label{kappa_value}
	Let $\gamma>1$, $\lambda>1$ and assume $(V(x),C(x))$ is a solution 
	of the similarity ODEs \eq{V_sim2}-\eq{C_sim2} satisfying \eq{condition1} 
	(with $\mu$ and $\nu$ finite and nonzero) and approaching $P_{+\infty}$ 
	as $x\downarrow -\infty$. Also, let $R\geq 0$ be given by \eq{entr_int} for 
	a constant $S>0$. 
	
	Then the conditions {\em (I)-(III)} of local integrability of conserved quantities, 
	together with pointwise boundedness as $r\downarrow0$ of the primary 
	flow variables $\rho$ and $c$ at times $t<0$, imply that the parameter
	$\kappa$ must take the value
	\beq\label{isentr_kappa}
		\kappa=-\textstyle\frac{2(\lambda-1)}{\gamma-1}.
	\eeq
\end{proposition}

\noindent {\bf Assumption:} It is assumed from now on that $\kappa$ takes the value in \eq{isentr_kappa}.

We note that \eq{isentr_kappa} gives $\alpha=0$, so that $a=2>0$ in \eq{ab}, verifying that $P_{+\infty}$ 
is a saddle point for the cases under consideration.
Also, with $\kappa$ given by \eq{isentr_kappa}, the 
integrability conditions (I)-(III) (recall that (II) is a consequence of (I) and (III))
reduce to the single constraint that
\beq\label{1st_lam_constr}
	\lambda<\bar\lambda(\gamma,n):=1+\textstyle\frac{n}{2}(1-\frac{1}{\gamma}).
\eeq
As \eq{isentr_kappa} yields a negative value for $\kappa$,
\eq{at_collapse} shows that also the density will blow up at collapse for the solutions 
under consideration.

\begin{remark}[Isentropic vs.\ full Euler self-similar solutions]\label{isentropic_comparison}
	Consider the simplified isentropic Euler model, i.e., 
	\eq{mass_m_d_full_eul}-\eq{mom_m_d_full_eul} with the constitutive relation
	$p\propto \rho^\gamma$. Radial self-similar flows of the form \eq{alt_sim_vars}, 
	suffering amplitude blowup in the presence of an everywhere positive pressure field, 
	have been constructed in \cites{jt3,mrrs1,mrrs2,b_c-l_g_s}. For the isentropic case the sound speed and density 
	are related according to $c^2\propto\rho^\gamma$, and it follows from \eq{alt_sim_vars}  
	that $\kappa$ takes the value \eq{isentr_kappa}. We therefore refer to 
	this as the {\em isentropic} $\kappa$-value.
		
	Conversely, for the full Euler system considered in the present work, it is immediate 
	to verify that with $\kappa$ given by \eq{isentr_kappa}, the relation \eq{entr_int} reduces to
	\beq\label{CR}
		(\textstyle\frac{C}{x})^2R^{1-\gamma}\equiv S \quad\text{(constant).}
	\eeq
	In terms of temperature $\theta\propto c^2$ and density $\rho$, this amounts to
	$\theta\rho^{1-\gamma}$ being constant, and it follows that the specific entropy 
	takes on a constant value within any region of continuity. Thus, the solutions 
	we build in this work describe
	isentropic flow away from the expanding shock wave generated at collapse. 
	This has the consequence that a substantial part of the analysis in 
	\cites{jt3,laz} carries over verbatim to the present case.
	
	However, we stress that the solutions we construct in this work 
	contain an expanding shock wave which propagates according to 
	the Rankine-Hugoniot conditions for the {\em full} Euler system. 
	In particular, the entropy suffers a jump across it, and the solution is 
	not globally isentropic.
	
	Examples of the various solution trajectories used in the construction are displayed 
	in Figure 1. Unsurprisingly, this is similar to the corresponding Figure 1 in 
	 \cite{jt3}. However, these figures are not identical: the Hugoniot curves 
	 $\Gamma_H$ are distinct in the two cases.
\end{remark}

For later reference we record the corresponding expression for the pressure field
when $\kappa$ takes the isentropic value in \eq{isentr_kappa}.
According to \eq{sound_speed}, \eq{alt_sim_vars},
and \eq{CR}, the pressure is then given by
\beq\label{pressure2}
	p(t,r)=\textstyle\frac{1}{\gamma}\rho(t,r) c^2(t,r)
	=\frac{S^\frac{1}{1-\gamma}}{\gamma\lambda^2}
	\left(r^{1-\lambda}\big|\frac{C(x)}{x}\big|\right)^\frac{2\gamma}{\gamma-1}.
\eeq

Next we record the Rankine-Hugoniot relations and entropy conditions
for a shock wave connecting two radial similarity flows of the form
\eq{alt_sim_vars}; see \cite{laz}. It is assumed that the parameters $\kappa$ and $\lambda$
are the same for both flows, and that the shock follows a path 
along which $x=\frac{t}{r^\lambda}\equiv x_\s$ (constant). 
With subscripts $0$ and $1$ referring to states immediately prior to and 
after passing through the shock, respectively, the Rankine-Hugoniot relations take the form:
\begin{align}
	1+V_1&=\textstyle\frac{\gamma-1}{\gamma+1}(1+V_0)+\frac{2C_0^2}{(\gamma+1)(1+V_0)}
	\label{rh1}\\
	C_1^2&=C_0^2+\textstyle\frac{\gamma-1}{2}[(1+V_0)^2-(1+V_1)^2]
	\label{rh2}\\
	R_1(1+V_1)&=R_0(1+V_0).
	\label{rh3}
\end{align}
(Note that neither $\lambda$, $\kappa$, nor $x_\s$ appears explicitly.)
We shall also need the entropy conditions for a shock wave of the third characteristic 
family: 3-characteristics should impinge on the shock forward in time \cite{daf}. 
Assuming that the shock propagates for $t>0$ (so that $x>0$), this implies that 
\beq\label{lax}
	0\geq C_0>-(1+V_0)\qquad\text{and}\qquad 0>-(1+V_1)>C_1,
\eeq
where we have also used $c\geq0$ and conservation of mass across the shock.
Thus, the state $(V_0,C_0)$ ($(V_1,C_1)$, repectively) corresponding 
to the flow at the immediate outside (inside, repectively) of a 3-shock propagating for $t>0$, must lie above
(below, repectively) the critical line $L_-=\{(V,-(1+V)\}$ in the $(V,C)$-plane. Our only use 
of these relations is for the expanding shock wave generated at the origin at time of 
collapse; see Figure 1 in which
the point $(V_0,C_0)$ is the lower end point of $\Gamma_3$ and $(V_1,C_1)$ (denoted $P_H$) 
is the upper end point of $\Gamma_4$ in the lower half of the $(V,C)$-plane.

\section{Construction of radial similarity Euler flows}\label{constr}

\subsection{The flow prior to collapse}\label{t<0}
As detailed in Remark \ref{isentropic_comparison}, the assumption that $\kappa$ 
takes the isentropic value \eq{isentr_kappa} renders the incoming flow isentropic. 
This circumstance allows us to use the analysis from \cites{jt3,laz} to build the 
flow prior to collapse, which amounts to determining 
the two trajectories $\Gamma_1$ and $\Gamma_2$ for \eq{CV_ode} in Figure 1.
More precisely, the value \eq{isentr_kappa} gives the similarity ODEs \eq{V_sim2}-\eq{C_sim2}
with $D(V,C)$ given by \eq{D} and
\begin{align}
	G(V,C)&=C^2(nV-\kappa)-V(1+V)(\lam+V)\label{G1}\\
	F(V,C)&=C\left\{C^2
	-k_1(1+V)^2+k_2(1+V)-k_3\right\}.\label{F1}
\end{align}
These yield precisely the similarity ODEs for the isentropic Euler model (cf.\ Eqns.\
(1.11)-(1.15) in \cite{jt3}). The analysis in \cites{jt3,laz} established 
that for $n=2$ or $3$ and $\gamma>1$, there is a number 
$\hat\lambda=\hat\lambda(\gamma,n)\in(1,\bar\lambda(\gamma,n))$ 
(with $\bar\lambda(\gamma,n)$ given in \eq{1st_lam_constr})
such that the following holds. For each $\lambda\in(1,\hat\lambda)$ there is a unique 
trajectory $\Gamma_1$ of \eq{CV_ode}  connecting $P_{+\infty}$ to a critical point $P_8\in \mathcal F\cap
\mathcal G\cap L_+$, which is a node for the parameter values under consideration. 
There are then  infinitely many trajectories $\Gamma_2$ of \eq{CV_ode} connecting 
$P_8$ to the star point (degenerate node) at $P_1=(0,0)$, and reaching the latter point with 
a finite, negative slope.  We fix any one of these 
as trajectory $\Gamma_2$. A typical configuration is displayed in Figure 1. 

We are free to fix an $x$-parametrization of the combined trajectory $\Gamma_1\cup\Gamma_2$
by choosing any $x_0<0$ and setting $(V(x_0),C(x_0))=P_8$.
Furthermore, it may be deduced from \eq{V_sim2}-\eq{C_sim2} that the points 
$P_1$ and $P_{+\infty}$ are approached as $x\uparrow0$ and $x\downarrow-\infty$, respectively. 
With this we obtain a continuous
solution $(V(x),C(x))$ of \eq{V_sim2}-\eq{C_sim2} defined for all $x<0$. (See Remark 3.1 in \cite{jt3} for details.) 
Finally, making a choice for the constant $S>0$ on the right-hand side of \eq{CR}, this provides, via 
\eq{alt_sim_vars}, a continuous solution $(\rho,u,c)(t,r)$ of the full, radial Euler system 
\eq{m_eul}-\eq{ener_eul}  for times $t<0$. 

In particular, it follows from the analysis in Section \ref{lam_kap_constrs}
that this solution is globally bounded in space at each time $\bar t<0$ prior to collapse, while it 
suffers amplitude blowup at the origin at time $t=0$. Finally, the combined trajectory 
$\Gamma_1\cup\Gamma_2$ of \eq{V_sim2}-\eq{C_sim2}
is contained within the half-strip $\{(V,C)\,|\, -1<V<0,\, C>0\}$, 
and satisfies $\frac{C(x)}{x}\to \mu\in(-\infty,0)$ as $x\uparrow 0$. Using this in \eq{pressure2}
shows that the resulting pressure field $p(t,r)$ is everywhere non-vanishing for $t<0$.


\begin{figure}
	\centering
	\includegraphics[width=13cm,height=15cm]{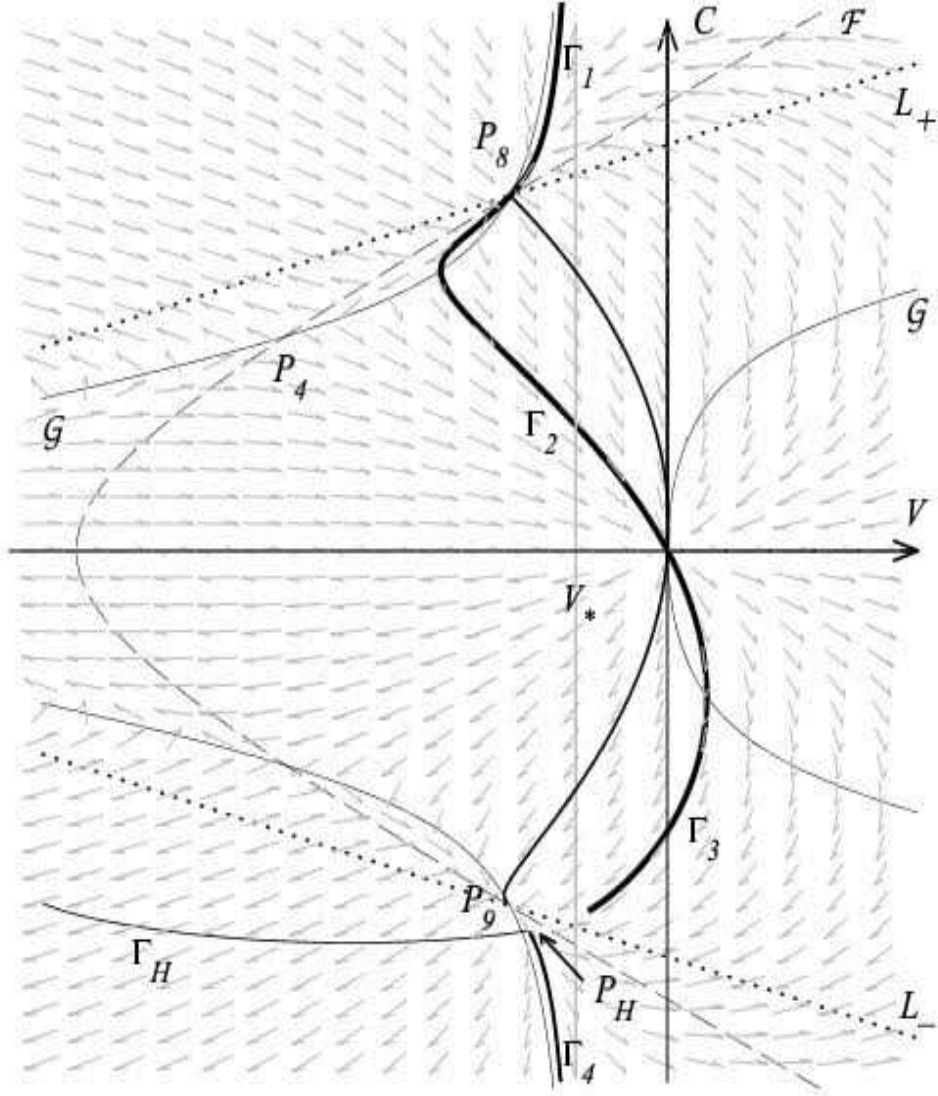}
	\caption{Maple plot with examples of solutions of \eq{CV_ode} with 
	$n=3$, $\gamma=3$, $\lambda=1.22$. The dotted curves are the critical lines $L_\pm$; the 
	zero levels $\mathcal F$ and $\mathcal G$ of $F$ and $G$ are the solid and dashed grey
	curves, respectively; the gray vertical line is $V=V_*$ (the asymptote of $\mathcal G$).
 	The thick solid black curve is the solution consisting of the four parts $\Gamma_1$-$\Gamma_4$;
	this is the type of solution constructed in Section \ref{constr}: Condition $(\Pi)$ is satisfied
	and the solution includes a jump from $\Gamma_3$ to $\Gamma_4$. The Hugoniot curve 
	$\Gamma_H$ corresponding to $\Gamma_3$ is the ``smiley'' curve in the third 
	quadrant which intersects $\Gamma_4$ at $P_H$.
	In addition, the thin solid black trajectory joining $P_8$ to $P_9:=(V_8,-C_8)$ illustrates
	how, for the same parameters $(n,\gamma,\lambda)$, Condition $(\Pi)$ may fail for 
	suitable choices of $\Gamma_2$. In such cases a complete, continuous solution of 
	\eq{CV_ode} is obtained by combining the latter trajectory with $\Gamma_1$, as well as with 
	the reflection of $\Gamma_1$ about the $V$-axis (i.e., the unique trajectory 
	joining $P_9$ to $P_{-\infty}$).}\label{Figure_2}
\end{figure} 

\subsection{The flow after collapse}\label{t>0}
We next continue the solution $(V(x),C(x))$ from Section \ref{t<0} through the origin
for $x>0$, entering the 4th quadrant of the $(V,C)$-plane with the same slope $\frac{\mu}{\nu}$
with which $\Gamma_2$ reached the origin (see Figure 1). The resulting trajectory of \eq{CV_ode} 
corresponding to $0<x<x_\s$ is denoted by $\Gamma_3$. Here, $x_\s$ is to be determined so that 
it gives the path (radius) $r_\s(t)=(\frac{t}{x_\s})^\frac{1}{\lambda}$ in physical space of an 
admissible 3-shock. The shock is required to connect $\Gamma_3$ to the unique trajectory $\Gamma_4$ of \eq{CV_ode}
which approaches the critical point $P_{-\infty}=(V_*,-\infty)$. Let $C=C_*(V)$ denote 
the unique solution of \eq{CV_ode} approaching $P_{+\infty}=(V_*,+\infty)$ (i.e., $C_*(V)$
parametrizes the trajectory $\Gamma_1$ determined above). 
It follows from the symmetries \eq{symms} that $\Gamma_4$ is given by $C=-C_*(V)$.

To determine $x_\s$ we argue as follows.
The trajectory $\Gamma_3$ starts out from the origin in the $(V,C)$-plane;
by continuity it follows that there is an $x'>0$ so that 
$(V(x),C(x))$ is located strictly above $L_-$ for $x\in(0,x')$. For each such $x$ we apply 
the Rankine-Hugoniot conditions \eq{rh1}-\eq{rh3} with $(V_0,C_0)=(V(x),C(x))$ to calculate 
the point $(V_1,C_1)=: (V_H(x),C_H(x))$ with the property that the latter 
point is the inside state of an admissible 3-shock with outside state $(V(x),C(x))$ (see below). 
We denote the curve $x\mapsto (V_H(x),C_H(x))$ by $\Gamma_H$, and refer to 
it as the Hugoniot curve corresponding to $\Gamma_3$ (see Figure 1). 

It is immediate from \eq{rh1}-\eq{rh2} that the starting point $(V_H(0),C_H(0))$ 
of $\Gamma_H$ is given by
\[(V_H(0),C_H(0))=\big(-\textstyle\frac{2}{\gamma+1},-\frac{\sqrt{2\gamma(\gamma-1)}}{\gamma+1}\big),\]
which is located in the 3rd quadrant and strictly below $L_-$. (Note that, while \eq{rh1}-\eq{rh2} only determines
$C_H(0)$ up to a sign, we must have $C_H(0)<0$ in order to satisfy the entropy condition \eq{lax}.)
Next, we claim that, as long as $(V_0,C_0)=(V(x),C(x))$ remains strictly above $L_-$, the 
point $(V_H(x),C_H(x))$ remains strictly below $L_-$. Indeed, $x\mapsto (V_H(x),C_H(x))$ is a 
continuous curve which, as just shown, starts out strictly below $L_-$. If, for contradiction, there were an 
$x$-value for which $(V_1,C_1)=(V_H(x),C_H(x))\in L_-$, then $C_1=-(1+V_1)$ and \eq{rh2} reduces to
\[\textstyle\frac{\gamma+1}{2}(1+V_1)^2=\textstyle\frac{\gamma-1}{2}(1+V_0)^2+C_0^2.\]
Dividing through by $(1+V_0)$ and using \eq{rh1} shows that $V_0=V_1$ in this case.
In turn, \eq{rh2} together with the requirement (by \eq{lax}) that  $\sgn C_0=\sgn C_1$ then yield 
$C_0=C_1$. Thus, if $\Gamma_H$ meets $L_-$,
then $\Gamma_3$ meets $L_-$ at the same point. 
Consequently, in using \eq{rh1}-\eq{rh2} to calculate $(V_1,C_1)=(V_H(x),C_H(x))$ (with $C_H(x)<0$)
from $(V_0,C_0)=(V(x),C(x))\in\Gamma_3$, the entropy conditions in \eq{lax} are met as long as 
$\Gamma_3$ remains above $L_-$.

The defining property of $x_\s$ is that the solution of \eq{CV_ode} starting from the point
$(V_H(x_\s),C_H(x_\s))$ is part of the trajectory $\Gamma_4$, the latter being the unique trajectory of 
\eq{CV_ode} which approaches the saddle point $P_{-\infty}$ as $x\uparrow+\infty$. That is, $x_\s$ is defined
by the requirement that 
\beq\label{x_s}
	-C_*(V_H(x_\s))=C_H(x_\s).
\eeq
We have found it necessary to assume that this equation has a solution corresponding to a shock wave:
\begin{itemize}
	\item Condition $(\Pi)$: The trajectory $\Gamma_4$ intersects the Hugoniot curve
	$\Gamma_H$ corresponding to $\Gamma_3$ at a point $P_H=(V_H(x_\s),C_H(x_\s))$ strictly
	below the critical line $L_-$.
\end{itemize}
In particular this implies that the discontinuity $(V(x_\s-),C(x_\s-))\mapsto P_H$
defines an admissible 3-shock.
(If $\Gamma_4$ were to intersect $\Gamma_H$ in more than one point  
below $L_-$, then any one of these will work for our purposes.)
\begin{remark}\label{Pi2}
	Numerical calculations provide robust evidence for the following (see Figure 1).
	For fixed  choices of $n=2$ or $3$ and $\gamma>1$, there are infinitely many 
	values of $\lambda\in(1,\hat\lambda(\gamma,n))$ for which:
	\begin{itemize}
		\item[(a)] there are infinitely many choices of $\Gamma_2$ (the trajectory
		joining $P_8$ to the origin) which generate 
		solutions satisfying Condition $(\Pi)$; and, at the same time:
		\item[(b)] there are infinitely many choices of $\Gamma_2$ which generate 
		solutions violating Condition $(\Pi)$ by having $\Gamma_4$ intersect $\Gamma_H$
		at $P_9:=(V_8,-C_8)$.
	\end{itemize}
	In the present work we assume that Condition $(\Pi)$ is satisfied. I.e., we 
	focus on scenario (a) which is exemplified by the trajectory 
	$\Gamma_1\cup\Gamma_2\cup\Gamma_3\cup\Gamma_4$ 
	(thick solid black curve) in Figure 1. The corresponding Euler flow contains
	an expanding shock wave emanating from the origin.
	
	In contrast, the other unlabeled (thin solid black) trajectory in Figure 1 connecting $P_8$ 
	to $P_9:=(V_8,-C_8)$ illustrates scenario (b). Taken 
	together with the trajectory $\Gamma_1$ and the reflection
	of $\Gamma_1$ about the $V$-axis (i.e., by \eq{symms}, 
	the unique trajectory connecting $P_9$ to $P_{-\infty}$), 
	we obtain a complete, continuous trajectory of \eq{CV_ode} connecting
	$P_{+\infty}$ to $P_{-\infty}$. The resulting Euler flow is continuous except at 
	$(t,r)=(0,0)$, where the flow suffers amplitude blowup. 
\end{remark}
Under the assumption that Condition $(\Pi)$ holds, we obtain $x_\s$ from \eq{x_s}, and 
then solve the similarity ODEs \eq{V_sim2}-\eq{C_sim2} for $x\in(x_\s,+\infty)$, 
with initial condition $(V(x_\s+),C(x_\s+))=P_H$. By construction this 
solution approaches $P_{-\infty}=(V_*,-\infty)$ as $x\uparrow+\infty$.

Next, we calculate $R(x_\s-)>0$ from $C(x_\s-)$ using \eq{CR} with the constant $S$ chosen 
in Section \ref{t<0}. By using the values of $V(x_\s-)$, $C(x_\s-)$, and $R(x_\s-)$ for $V_0$, $C_0$, and $R_0$ in
the Rankine-Hugoniot relations \eq{rh1}-\eq{rh3}, we obtain the values $V_1=V_H(x_\s)$,
$C_1=C_H(x_\s)$ (already determined above), 
as well as $R_1$. (Note that $R_0>0$ together with the entropy conditions in \eq{lax} yield
$R_1>0$.) The latter two values determine, via \eq{CR}, the value 
\[S'=(\textstyle\frac{C_1}{x_\s})^2R_1^{1-\gamma}>0\]
and thus the value of the specific entropy in the region within the expanding shock (see Remark \ref{isentropic_comparison}).
It is immediate to verify that \eq{pressure2} again gives an everywhere strictly 
positive pressure field in the flow after collapse.
This concludes the construction of the relevant solution $(V(x),C(x))$ of \eq{V_sim2}-\eq{C_sim2},
and yields (via \eq{alt_sim_vars} and \eq{entr_int}) the globally defined similarity 
solutions $(\rho(t,r),u(t,r),c(t,r))$ described in Main Result \ref{main_result}.

%
%
%

Finally, we record the following estimates that are consequences of the construction above,
and which are used repeatedly in the analysis below:
\beq\label{VC_props}
	\textstyle\left|\frac{C(x)}{x}\right|,\, \left|\frac{V(x)}{x}\right|\lesssim 1
	\quad\text{for $|x|\leq 1$, and }\quad
	\left|\frac{C(x)}{x}\right|\lesssim |x|^{\frac{1}{\lambda}-1},\, |V(x)|\lesssim 1
	\quad\text{for $|x|\geq 1$.}
\eeq

\section{Radial similarity flows as weak solutions}\label{weak_solns}
Given a radial similarity flow $(\rho(t,r),u(t,r),c(t,r))$ as constructed in Section \ref{constr},
we define (with a slight abuse of notation) the multi-d quantities
\beq\label{assmbld}
	\rho(t,{\bf x}):=\rho(t,|{\bf x}|),\qquad 
	{\bf u}(t,{\bf x}):=u(t,|{\bf x}|)\textstyle\frac{{\bf x}}{|{\bf x}|},\qquad 
	e(t,{\bf x}):=e(t,|{\bf x}|),
\eeq
where (see \eq{sound_speed})
\beq\label{ec}
	e(t,r)=\textstyle\frac{1}{\gamma(\gamma-1)}c^2(t,r).
\eeq
It remains to show that these quantities define a genuine weak solution of the original, 
multi-d Euler system \eq{mass_m_d_full_eul}-\eq{energy_m_d_full_eul}
according to the following definition.
(We write $\rho(t)$ for $\rho(t,\cdot)$ etc., and ${\bf u}=(u_1,\dots,u_n)$.)
\begin{definition}\label{weak_soln}
	Consider the compressible Euler system 
	\eq{mass_m_d_full_eul}-\eq{energy_m_d_full_eul}
	in $n$ space dimensions with pressure function $p=p(\rho,e)$.
	The measurable functions $\rho,\, u_1,\dots,u_n,e:\RR\times \RR^n\to \RR$,
	with $\rho\geq0$ and $e\geq 0$, 
	constitute a {\em weak solution} of \eq{mass_m_d_full_eul}-\eq{energy_m_d_full_eul}
	provided that:
	\begin{itemize}
		\item[(A)] the maps $t\mapsto \rho(t)$, $t\mapsto \rho(t) u(t)$, and 
		$t\mapsto \rho(t)\left(e(t)+ \frac{1}{2}u^2(t)\right)$ 
		belong to $C^0(\RR;L^1_{loc}(\RR^n))$;
		\item[(B)]  the functions $\rho u^2$, $p$, and $\left[\rho\left(e+ \frac{1}{2}u^2\right)+p\right] u$
		belong to $L^1_{loc}(\RR\times\RR^n)$;
		\item[(C)] the conservation laws for mass, momentum, and energy are 
		satisfied weakly in sense that
		\begin{align}
			\int_\RR\int_{\RR^n} \rho\vp_t+\rho{\bf u}\cdot\nabla_{\bf x}\vp
			\, d{\bf x}dt &=0\label{m_d_mass_weak}\\
			\int_\RR\int_{\RR^n} \rho u_i\vp_t
			 +\rho u_i{\bf u}\cdot\nabla_{\bf x}\vp+p\vp_{x_i}\, d{\bf x}dt &=0
			 \qquad \text{for $i=1,\dots, n$,} \label{m_d_mom_weak}\\
			\int_\RR\int_{\RR^n} \rho\big(e+\textstyle \frac{1}{2}u^2\big)\vp_t
			 +\left[\rho\left(e+ \frac{1}{2}u^2\right)+p\right] {\bf u}\cdot\nabla_{\bf x}\vp\, d{\bf x}dt &=0
			 \label{m_d_energy_weak}
		\end{align}
		whenever $\vp\in C_c^1(\RR\times \RR^n)$ (the set of $C^1$ functions with compact support).
	\end{itemize}
\end{definition}
\begin{remark}
	Condition {\em (A)} requires that the conserved quantities define
	continuous maps into $L^1_{loc}(\RR^n)$, which is the natural function
	space in this setting. Conditions {\em (A)} and {\em (B)} ensure that all terms occurring 
	in the weak formulations \eq{m_d_mass_weak}-\eq{m_d_energy_weak} are 
	locally integrable in space and time.
\end{remark}
We first rewrite Definition \ref{weak_soln} for radial solutions.
Set $r=|\bf x|$, $m=n-1$, $\RR_0^+=[0,\infty)$,
\[ L^1_{loc}(r^mdr)=L^1_{loc}(\RR^+_0, r^mdr),\qquad
\text{and}\qquad L^1_{loc}(dt\times r^mdr)=L^1_{loc}(\RR\times\RR^+_0,dt\times r^mdr).\]
It is convenient to introduce the following (nonstandard) notation:
$C^1_c(\RR\times\RR^+_0)$ denotes the set of real-valued $C^1$ functions 
$\psi:\RR\times\RR^+_0\to\RR$ that vanish outside $[-\bar t,\bar t]\times[0,\bar r]$, 
for some $\bar t,\, \bar r\in\RR^+$. Finally, $C^1_0(\RR\times\RR^+_0)$ denotes 
the set of those $\psi\in C^1_c(\RR\times\RR^+_0)$ that also  
satisfy $\psi(t,0)\equiv 0$.

\begin{definition}\label{rad_symm_weak_soln}
	With the same setup as in Definition \ref{weak_soln}, 
	the measurable functions $\rho,\, u,\, e:\RR\times \RR^+_0\to \RR$,
	with $\rho\geq0$ and $e\geq0$,  
	constitute a {\em radial weak solution} of \eq{m_eul}-\eq{ener_eul} provided that:
	\begin{itemize}
		\item[(i)]  the maps $t\mapsto \rho(t)$, $t\mapsto \rho(t) u(t)$, and 
		$t\mapsto \rho(t)\left(e(t)+ \frac{1}{2}u^2(t)\right)$ 
		belong to $C^0(\RR;L^1_{loc}(r^mdr))$;
		\item[(ii)]  the functions $\rho u^2$, $p$, and $\left[\rho\left(e+ \frac{1}{2}u^2\right)+p\right] u$
		belong to $L^1_{loc}(dt\times r^mdr)$;
		\item[(iii)] the conservation laws for mass, momentum, and energy are 
		satisfied in the sense that
		\begin{align}
			\int_{\RR}\int_{\RR^+} \left(\rho\psi_t+\rho u\psi_r\right) r^mdrdt &=0 
			\qquad\forall \psi\in C^1_c(\RR\times\RR^+_0) \label{radial_mass_weak}\\
			\int_{\RR}\int_{\RR^+} \left(\rho u\psi_t
			+\rho u^2\psi_r+p\big(\psi_r+\textstyle\frac{m\psi}{r}\big)\right) r^mdrdt &=0 
			\qquad\forall \psi\in C^1_0(\RR\times\RR^+_0)\label{radial_mom_weak}\\
			\int_{\RR}\int_{\RR^+} \Big(\rho\big(e+\textstyle \frac{1}{2}u^2\big)\psi_t
			 +\left[\rho\left(e+ \frac{1}{2}u^2\right)+p\right] u\psi_r\Big) r^mdrdt &=0 
			\qquad\forall \psi\in C^1_c(\RR\times\RR^+_0).\label{radial_ener_weak}
		\end{align}
	\end{itemize}
\end{definition}
\noindent We record the fact that the latter definition is consistent with the former one:
\begin{proposition}\label{hoff_prop}
	Assume that $(\rho(t,r),u(t,r),e(t,r))$ is a radial weak solution of  
	\eq{m_eul}-\eq{ener_eul} according to Definition 
	\ref{rad_symm_weak_soln}. Then the corresponding multi-d 
	quantities defined according to \eq{assmbld} provide a weak solution 
	of the multi-d system  \eq{mass_m_d_full_eul}-\eq{energy_m_d_full_eul}
	according to Definition \ref{weak_soln}.
\end{proposition}
\begin{proof}
	This was established in \cite{jt1} (Proposition V.I). 
\end{proof}
For a given similarity flow $(\rho(t,r),u(t,r),c(t,r))$ constructed as in Section \ref{constr},
and with $e$ and $p$ defined by \eq{ec} and  \eq{pressure1}, respectively, 
we proceed to verify the conditions in Definition \ref{rad_symm_weak_soln}. 
We recall that the construction in Section \ref{constr} was carried out under the 
following conditions on the parameters $\kappa$ and $\lambda$:
\beq\label{params}
	\kappa=-\textstyle\frac{2(\lambda-1)}{\gamma-1}, \qquad 
	1<\lambda<\bar\lambda(\gamma,n):=1+\textstyle\frac{n}{2}(1-\frac{1}{\gamma});
\eeq
see Sections \ref{lam_kap_constrs} and \ref{t<0}. To simplify notation we let $B$ denote the 
$x$-value corresponding to the expanding shock (i.e., $B=x_\s$). Also, note that $\rho e\propto p\propto \rho c^2$,
and that the density is given via \eq{CR} as
\beq\label{rho_c}
	\rho(t,r)=r^\kappa R(t/r^\lambda)\propto r^\kappa
	\textstyle\left|\frac{C(t/r^\lambda)}{t/r^\lambda}\right|^\frac{2}{\gamma-1}.
\eeq

Concerning the conditions in Definition \ref{rad_symm_weak_soln}, 
non-negativity of $\rho$ and $e$ hold by construction. The remaining conditions 
(i)-(iii) of Definition \ref{rad_symm_weak_soln} are analyzed as follows.

\subsection{Integrability and continuity}\label{int_cont}
\begin{lemma}\label{L^1_loc}
	With the setup above we have that the maps $t\mapsto \rho(t)|u(t)|^k$, for $k=0,1,2$, and $t\mapsto \rho(t)c(t)^2$
	all take values in $L^1_{loc}(r^mdr)$.
\end{lemma}
\begin{proof}
	Fix $\bar r>0$ and any $t\in \RR$. Consider first the case that $t\neq 0$. 
	Using \eq{alt_sim_vars} and \eq{rho_c} we have
	\beq\label{rho_1}
		\|\rho(t)u(t)^k\|_{L^1((0,\br);r^mdr)} 
		\lesssim \int_0^{\br} r^{\kappa+m+k(1-\lambda)}
		\textstyle\left|\frac{C(t/r^\lambda)}{t/r^\lambda}\right|^\frac{2}{\gamma-1}
		\left|\frac{V(t/r^\lambda)}{t/r^\lambda}\right|^k\, dr,
	\eeq
	and 
	\beq\label{c_1}
		\|\rho(t)c(t)^2\|_{L^1((0,\br);r^mdr)} 
		\lesssim \int_0^{\br} r^{\kappa+m+2(1-\lambda)}
		\textstyle\left|\frac{C(t/r^\lambda)}{t/r^\lambda}\right|^\frac{2\gamma}{\gamma-1}\, dr,
	\eeq
	According to \eq{VC_props} we have that both $|V(x)|$ and $\left|\frac{C(x)}{x}\right|$ are globally 
	bounded, with the latter satisfying
	\beq\label{C_est}
		\left|\textstyle\frac{C(x)}{x}\right|\lesssim x^{\frac{1}{\lambda}-1}.
	\eeq
	Using this in \eq{rho_1} and \eq{c_1} gives
	\[\|\rho(t)u(t)^k\|_{L^1((0,\br);r^mdr)} 
	\lesssim |t|^{\frac{\kappa}{\lambda}-k} \int_0^{\br} r^{m+k}\, dr<\infty\qquad\text{for $t\neq0$,}\]
	and
	\[\|\rho(t)c(t)^2\|_{L^1((0,\br);r^mdr)} 
	\lesssim |t|^\frac{\kappa\gamma}{\lambda} \int_0^{\br} r^{m}\, dr<\infty\qquad\text{for $t\neq0$.}\]
	These expressions are not useful at $t=0$ (since $\kappa<0$); instead we use that 
	$V(x)$ and $C(x)$ were constructed to meet the conditions in 
	\eq{condition1} (with $\nu$ and $\mu$ being finite real numbers; see Section \ref{t<0}).
	The flow variables are therefore given by \eq{at_collapse}, which yields
	\[\|\rho(0)u(0)^k\|_{L^1((0,\br);r^mdr)} 
	\propto \int_0^{\br} r^{\kappa+m+k(1-\lambda)}\, dr\qquad (k=0,1,2).\]
	It is immediate to verify the following: these integrals are finite if and only if the 
	integral with $k=2$ is finite; this condition is the same as that guaranteeing
	finiteness of $\|\rho(0)c(0)^2\|_{L^1((0,\br);r^mdr)}$; and it holds if and only if 
	$\lambda<\bar\lambda(\gamma,n)$. The latter condition holds by assumption
	(see \eq{params}), concluding the proof.
\end{proof}
We next establish the continuity of these maps.
\begin{lemma}\label{cont_into_L^1_loc}
	With the setup above we have that the maps $t\mapsto \rho(t)|u(t)|^k$, for $k=0,1,2$, 
	and $t\mapsto \rho(t)c(t)^2$ all belong to $C^0(\RR;L^1_{loc}(r^mdr))$.
\end{lemma}
\begin{proof}
	Since the argument is similar for all cases we give the details 
	only for the map $t\mapsto \rho(t)c(t)^2$. Fix $\br>0$ and any $t\in\RR$. For
	$s,t\neq 0$, \eq{alt_sim_vars} and \eq{rho_c} give
	\[\|\rho(t)c(t)^2-\rho(s)c(s)^2\|_{L^1((0,\br);r^mdr)} \lesssim \int_0^{\br} r^{\kappa+m+2(1-\lambda)}
	\textstyle\left|\left|\frac{C(t/r^\lambda)}{t/r^\lambda}\right|^\frac{2\gamma}{\gamma-1}
	-\left|\frac{C(s/r^\lambda)}{s/r^\lambda}\right|^\frac{2\gamma}{\gamma-1}\right|\, dr.\]
	Since $x\mapsto C(x)$ is continuous almost everywhere, the integrand
	tends pointswise a.e.\ to zero as $s\to t$. Applying \eq{C_est} we get that 
	the integrand in the last integral is bounded by
	\beq\label{intermediate}
		r^{\kappa+m+2(1-\lambda)} (|t|^\frac{\kappa\gamma}{\lambda}+
		|s|^\frac{\kappa\gamma}{\lambda})r^{-\kappa\gamma}
		= (|t|^\frac{\kappa\gamma}{\lambda}+|s|^\frac{\kappa\gamma}{\lambda})r^m,
	\eeq
	which is uniformly bounded in $L^1((0,\br);dr)$ as $s\to t\neq0$.
	An application of the Dominated Convergence Theorem therefore gives 
	$\rho(s)u(s)^k\to \rho(t)u(t)^k$ in $L^1((0,\br);r^mdr)$ when $s\to t\neq0$.

	Due to the negative time exponent $\frac{\kappa\gamma}{\lambda}$ in \eq{intermediate} 
	a slightly different argument is required for continuity at time $t=0$. 
	Using \eq{condition1}, \eq{at_collapse}, \eq{alt_sim_vars}, and \eq{rho_c} we have
	\beq\label{dct_0}
		\|\rho(0)c(0)^2-\rho(s)c(s)^2\|_{L^1((0,\br);r^mdr)} 
		\sim \int_0^{\br} r^{\kappa+m+2(1-\lambda)}
		\textstyle\left|\mu^\frac{2\gamma}{\gamma-1}
		-\left|\frac{C(s/r^\lambda)}{s/r^\lambda}\right|^\frac{2\gamma}{\gamma-1}\right|\, dr.
	\eeq
	Since $\frac{C(x)}{x}\to\mu$ as $x\to 0$, the integrand in \eq{dct_0} tends pointwise to zero
	as $s\to 0$. To uniformly bound the integrand we recall \eq{VC_props}:
	\[\textstyle\left|\frac{C(x)}{x}\right|\lesssim 1
		\quad\text{for $|x|\leq 1$, and}\qquad
		\left|\frac{C(x)}{x}\right|\lesssim |x|^{\frac{1}{\lambda}-1}
		\qquad\text{for $|x|\geq 1$.}\]
Using this in \eq{dct_0} we get that, for $|s|\leq \br^\lambda$, the integrand is bounded (up to a multiplicative constant) by 
\begin{align*}
r^{\kappa+m+2(1-\lambda)}
\textstyle\left(1+\chi_{[0,|s|^\frac{1}{\lambda}]}(r)
\left(\frac{|s|}{r^\lambda}\right)^{\frac{\kappa\gamma}{\lambda}} \right)
\lesssim r^{\kappa+m+2(1-\lambda)},
\end{align*}
where in the last step we have used that $\kappa<0$. Finally, as noted in the proof of 
Lemma \ref{L^1_loc}, the assumption that $\lambda<\bar\lambda(\gamma,n)$ 
is equivalent to $\kappa+m+2(1-\lambda)>-1$, showing that the integrand in 
\eq{dct_0} is bounded by a fixed $L^1((0,\br);dr)$-function for all $|s|\leq \br^\lambda$. 
The Dominated Convergence Theorem therefore
gives that $\rho(s)c(s)^2\to \rho(0)c(0)^2$ in $L^1((0,\br);r^mdr)$ as $s\to 0$.
\end{proof}
As $\rho e\propto \rho c^2$, Lemma \ref{cont_into_L^1_loc} establishes condition (i) 
of Definition \ref{rad_symm_weak_soln}. By continuity, and since $p\propto \rho c^2$, 
it also provides a part of condition (ii):
\begin{corollary}\label{L^1_loc_t_x_1}
	With the setup above, the functions $\rho u^2$ and $p$ belong to $L^1_{loc}(dt\times r^mdr)$.
\end{corollary}
For condition (ii) it remains to verify that the function $\left[\rho\left(e+ \frac{1}{2}u^2\right)+p\right]\! u$
belongs to $L^1_{loc}(dt\times r^mdr)$. (Note that this condition is not present in the isentropic 
case considered in \cite{jt3}.) Again, since $\rho e,p\propto \rho c^2$, this follows once we establish 
the following:
\begin{lemma}\label{L^1_loc_t_x_2}
	With the setup above, the functions $\rho u^3$ and $\rho c^2 u$ belong to 
	$L^1_{loc}(dt\times r^mdr)$.
\end{lemma}
\begin{proof}
	We present the details for the representative case of $\rho c^2 u$. Recall that, by construction 
	(see Section \ref{constr}) the functions $V(x)$ and $C(x)$ suffer a discontinuity
	at $x=B>0$, corresponding to the expanding shock wave generated at $t=0$. To establish the 
	claim it suffices to show that 
	\[I:=\int_{-T}^T\int_0^\br \rho c^2 u r^m\, drdt<\infty\]
	whenever $\br>0$ and $T=B\br^\lambda$. According to \eq{alt_sim_vars} and \eq{rho_c} 
	we have
	\[I\lesssim \int_{-T}^T\int_0^\br r^{\kappa+m+3(1-\lambda)}
	\left|\frac{C(x)}{x}\right|^\frac{2\gamma}{\gamma-1}\left|\frac{V(x)}{x}\right|\, drdt,\]
	where $x=\frac{t}{r^\lambda}$. Changing integration variable from $r$ to $x$ and splitting
	into two parts corresponding to $t\gtrless 0$, yield
	\[I\lesssim \left\{\int_{-T}^0\int_{-\infty}^{t/\br^\lambda} 
	+ \int_0^T\int_{t/\br^\lambda}^{\infty} \right\}r^{\kappa+m+3(1-\lambda)}
	\left|\frac{C(x)}{x}\right|^\frac{2\gamma}{\gamma-1}\left|\frac{V(x)}{x}\right|
	\frac{|t|^\frac{1}{\lambda}}{|x|^{1+\frac{1}{\lambda}}}\, dxdt=:I_1+I_2,\]
	where $r=|\frac{t}{x}|^\frac{1}{\lambda}$. According to \eq{VC_props} $V(x)$ and $C(x)$ 
	satisfy the same bounds at $x=\pm\infty$ and also at $x=0\pm$. It therefore suffices to 
	consider $I_2$. Substituting for $r$ and splitting the $x$-integration into $x\gtrless B$,
	we have (recall that $T=B\br^\lambda$)
	\[I_2= \int_0^T t^\frac{\kappa+n+3(1-\lambda)}{\lambda}\left(
		\left\{\int_{t/\br^\lambda}^B+\int_{B}^{\infty}\right\}
		\left|\frac{C(x)}{x}\right|^\frac{2\gamma}{\gamma-1}\left|\frac{V(x)}{x}\right|
		\, \frac{dx}{|x|^{1+\frac{\kappa+n+3(1-\lambda)}{\lambda}}}\right)dt=:J_1+J_2.\]
	Applying the bounds in \eq{VC_props} we have
	\[J_1\lesssim \int_0^T t^\frac{\kappa+n+3(1-\lambda)}{\lambda}\left(
		\int_{t/\br^\lambda}^B\frac{dx}{|x|^{1+\frac{\kappa+n+3(1-\lambda)}{\lambda}}}\right)dt
		\lesssim 1+ \int_0^T t^\frac{\kappa+n+3(1-\lambda)}{\lambda}\, dt.\]
	The last integral is bounded provided $\frac{\kappa+n+3(1-\lambda)}{\lambda}>-1$,
	and a direct calculation shows that this is a consequence of the standing assumption
	that $\lambda<\bar\lambda(\gamma,n)$. Thus $J_1<\infty$. For $J_2$, the bounds in 
	\eq{VC_props} yield 
	\[J_2\lesssim \int_0^Tt^\frac{\kappa+n+3(1-\lambda)}{\lambda}\left(
	\int_B^\infty x^\frac{\kappa(\gamma-1)-2\lambda-n-3(1-\lambda)}{\lambda}\, dx\right)dt.\]
	It is straightforward to verify that $\frac{\kappa(\gamma-1)-2\lambda-n-3(1-\lambda)}{\lambda}<-1$
	(whenever $\lambda>0$),
	so that the last $x$-integral is bounded. The finiteness of $J_2$ now follows as for $J_1$.
	This shows that $I<\infty$, so that $\rho c^2 u\in L^1_{loc}(dt\times r^mdr)$.
\end{proof}
With Lemma \ref{L^1_loc_t_x_2} conditions (i) and (ii) of Definition \ref{rad_symm_weak_soln}
are verified for the radial similarity solutions constructed in Section \ref{constr}.

\subsection{Weak forms of the equations}\label{weak_forms}
The strategy for verifying the weak forms of the conservation laws
is the same for each of \eq{radial_mass_weak}-\eq{radial_ener_weak}. 
Also, since $\kappa$ takes the isentropic value in \eq{params}, the arguments
for the mass and the momentum equations \eq{radial_mass_weak}-\eq{radial_mom_weak}
reduce to those given in \cite{jt3}. We therefore present the details only for
verifying \eq{radial_ener_weak}.

For this, fix a test function $\psi\in C^1_c(\RR\times\RR^+_0)$ with 
$\supp\psi\subset[-T,T]\times [0,\br]$. By increasing $\br$ or 
$T$ if necessary we may assume without loss of generality that $T=B\br^\lambda$.
Next, for any $\delta\in(0,\br)$ we define the open regions (see Figure 2)
\[J_\delta=\left\{(t,r)\,|\, -T<t<T,\, \delta<r<\br,\, \textstyle\frac{t}{r^\lambda}<B \right\},\]
and
\[K_\delta=\left\{(t,r)\,|\, 0<t<T,\, \delta<r<\br,\,\textstyle\frac{t}{r^\lambda}>B \right\}.\] 
With $E:=e+ \frac{1}{2}u^2$ and
\begin{align}
	\mathcal E(\psi)&:=\iint_{\RR\times\RR^+} \left[\rho E\psi_t+(\rho E+p)u\psi_r\right]\, r^mdrdt \nn\\
	&= \Big\{\iint_{\RR\times [0,\delta]} +\iint_{J_\delta} +\iint_{K_\delta}\Big\}
	\left[\rho E\psi_t+(\rho E+p)u\psi_r\right]\,  r^mdrdt \nn\\
	&=:\mathcal E_\delta(\psi)
	+\Big\{\iint_{J_\delta} +\iint_{K_\delta}\Big\}
	\left[\rho E\psi_t+(\rho E+p)u\psi_r\right]\,  r^mdrdt,
	\label{E_psi}
\end{align}
the claim is that $\mathcal E(\psi)$ vanishes. We shall establish this by showing 
that the right hand side of \eq{E_psi} tends to zero as $\delta\downarrow 0$. 
\begin{figure}\label{Figure_3}
	\centering
	\includegraphics[width=7cm,height=9cm]{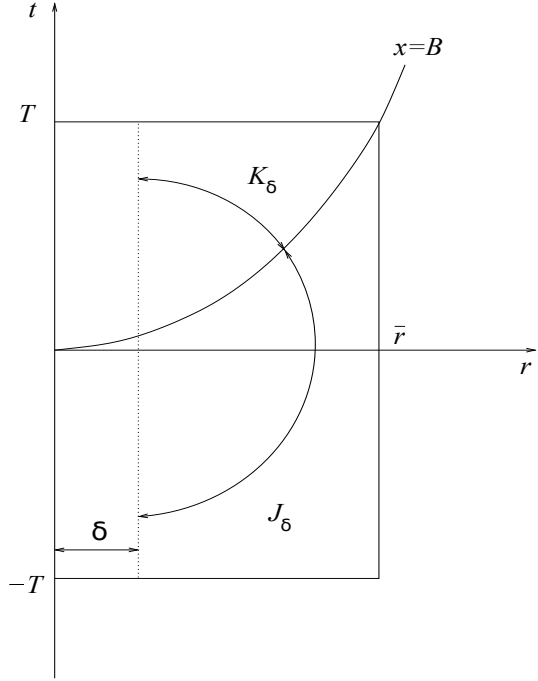}
	\caption{Regions of integration for verifying the weak form of the conservation laws.}
\end{figure} 
First, according to the properties established in Section \ref{int_cont}
the integrand of $\mathcal E(\psi)$ is locally integrable. It follows that  
$\mathcal E_\delta(\psi)$ tends to zero as $\delta\downarrow 0$. 

It remains to estimate the integrals over $J_\delta$ and $K_\delta$ in \eq{E_psi}. 
For this we use that $(\rho,u,c)$, by construction, is a 
classical (Lipschitz continuous) solution of the full Euler system \eq{m_eul}-\eq{ener_eul}
within each of the open regions $J_\delta$ and $K_\delta$. In particular, it follows that the 
energy equation
\[(\rho E r^m)_t+[(\rho E+p)u r^m]_r=0\]
is satisfied classically in each of $J_\delta$ and $K_\delta$. 
Also by construction, the Rankine-Hugoniot jump relations are satisfied 
across their common boundary along the curve 
$r=(\frac{t}{B})^\frac{1}{\lambda}$. 
Applying the divergence theorem to each region we therefore obtain
\beq\label{part1}
	\Big\{\iint_{J_\delta} +\iint_{K_\delta}\Big\}\left[\rho E\psi_t+(\rho E+p)u\psi_r\right]\, r^mdrdt
	=\delta^m\int_{-T}^T[(\rho E+p)u\psi](t,\delta)\, dt=:\mathcal I_\delta.
\eeq
Using \eq{alt_sim_vars} and \eq{rho_c}, and changing integration variable from $t$ to $x$, 
we get that
\[\left|\mathcal I_\delta\right|
\lesssim \delta^{m+\lambda+\kappa+3(1-\lambda)}
\left\{\int_{|x|\leq 1}+\int_{1\leq|x|\leq\frac{T}{\delta^\lambda}}\right\}
\textstyle\left|\frac{C(x)}{x}\right|^\frac{2}{\gamma-1}\left|\frac{V(x)}{x}\right|
\left(\left|\frac{V(x)}{x}\right|^2+\left|\frac{C(x)}{x}\right|^2\right)\, dx,\]
where we have split up the $x$-integration into two parts (assuming  $\delta<T^\frac{1}{\lambda}$).
Making use of \eq{VC_props}, and in particular that $\left|\frac{C(x)}{x}\right|$ 
dominates $\left|\frac{V(x)}{x}\right|$ for large values of $|x|$, we obtain 
\beq\label{1st_part}
	\left|\mathcal I_\delta\right|\lesssim 
	\delta^{n+\kappa+2(1-\lambda)}\left\{1+\int_1^{T/\delta^\lambda}x^{\frac{\kappa\gamma}{\lambda}-1}\, dx\right\}
	\lesssim \delta^{n+\kappa+2(1-\lambda)}\left\{1+\delta^{-\kappa\gamma}\right\}
	\lesssim \delta^{n+\kappa+2(1-\lambda)},
\eeq
where in the last step we have used that $\kappa<0$ and $\delta$ is small. 
Finally, it is immediate to verify that the condition $n+\kappa+2(1-\lambda)>0$ 
is equivalent to the standing requirement $\lambda<\bar\lambda(\gamma,n)$. 
This shows that the right-hand side of \eq{E_psi} vanishes as $\delta\downarrow 0$,
so that the weak form of the energy equation \eq{radial_ener_weak} is satisfied. 
\qed

This verifies condition (iii) of Definition \ref{rad_symm_weak_soln} and 
concludes the demonstration of our Main Result.

\begin{remark}
        It follows from Lemma \ref{L^1_loc} that the solutions under consideration
        have locally finite mass, momentum, and energy at all times. However, their 
        total mass, momentum, and energy are unbounded; e.g., at time $t=0$, this 
        is an immediate consequence of \eq{at_collapse} and the restrictions \eq{params}
        on the values of $\kappa$ and $\lambda$.
        On the other hand, as in the isentropic case \cite{jt3},
        the solutions can be altered outside a bounded set so as
        to give examples of amplitude blowup from initial data with finite mass, 
        momentum, and energy. This is a consequence of the fact that 1-characteristics 
        starting at negative times from points along $x\equiv \bar x<0$ cross the
        $r$-axis at strictly positive locations at time of collapse. It also follows from this
        that the same holds for particle trajectories. In particular, mass does not
        ``accumulate'' at the origin and there is never a Dirac distribution present in the density field.
\end{remark}

\bigskip
\paragraph{\bf Acknowledgment.}
This material is based in part upon work supported by the National Science Foundation under Grant Numbers DMS-1813283 (Jenssen) and DMS-1714912 (Tsikkou). Any opinions, findings, and conclusions or recommendations expressed in this material are those of the authors and do not necessarily reflect the views of the National Science Foundation.


\end{document}